\documentclass[conference,a4paper]{IEEEtran}

\usepackage[utf8]{inputenc}
\usepackage[english]{babel}
\usepackage[T1]{fontenc}
\usepackage{epsfig}
\usepackage[cmex10]{amsmath}
\usepackage{amssymb, amsbsy, amsthm}
\usepackage{mathrsfs}
\usepackage{xspace}
\usepackage{color}
\usepackage{url}
\usepackage{booktabs}
\usepackage[linesnumbered,ruled,vlined,titlenumbered]{algorithm2e}
\usepackage{romannum}
\usepackage{tikz}
\usepackage{cite}
\IEEEoverridecommandlockouts

\newtheorem{theorem}{Theorem}

\newtheorem{definition}{Definition}
\newtheorem{example}{Example}

\newtheorem{problem}{Problem}

\usepackage{varioref}
\usepackage{fancyref}

\makeatletter
\def\mkfancyprefix#1#2{%
\expandafter\def\csname fancyref#1labelprefix\endcsname{#1}%
\begingroup\def\x{\endgroup\frefformat{plain}}%
    \expandafter\x\csname fancyref#1labelprefix\endcsname
    {\MakeLowercase{#2}\fancyrefdefaultspacing##1}%
\begingroup\def\x{\endgroup\Frefformat{plain}}%
    \expandafter\x\csname fancyref#1labelprefix\endcsname
    {#2\fancyrefdefaultspacing##1}%
\begingroup\def\x{\endgroup\frefformat{vario}}%
    \expandafter\x\csname fancyref#1labelprefix\endcsname
    {\MakeLowercase{#2}\fancyrefdefaultspacing##1##3}%
\begingroup\def\x{\endgroup\Frefformat{vario}}%
    \expandafter\x\csname fancyref#1labelprefix\endcsname
    {#2\fancyrefdefaultspacing##1##3}%
}
\makeatother
\fancyrefchangeprefix{\fancyrefeqlabelprefix}{eqn}
\mkfancyprefix{sec}{Section}
\mkfancyprefix{subsec}{Section}
\mkfancyprefix{tbl}{Table}
\mkfancyprefix{thm}{Theorem}
\mkfancyprefix{lem}{Lemma}
\mkfancyprefix{cor}{Corollary}
\mkfancyprefix{prop}{Proposition}
\mkfancyprefix{prob}{Problem}
\mkfancyprefix{alg}{Algorithm}
\mkfancyprefix{inv}{Invariant}
\mkfancyprefix{ex}{Example}
\mkfancyprefix{line}{Line}
\mkfancyprefix{def}{Definition}
\mkfancyprefix{itm}{Item}
\newcommand{\cref}[1]{\Fref{#1}}

\makeatletter
\newcommand{\removelatexerror}{\let\@latex@error\@gobble}
\makeatother

\newcommand{\printalgoIEEE}[1]
{{\centering
\scalebox{0.97}{
\removelatexerror
\begin{tabular}{p{\columnwidth}}
\begin{algorithm}[H]
 \begin{small}
 #1
 \end{small}
\end{algorithm}
\end{tabular}
}
}
}

\def\ve#1{{\mathchoice{\mbox{\boldmath$\displaystyle #1$}}%
              {\mbox{\boldmath$\textstyle #1$}}%
              {\mbox{\boldmath$\scriptstyle #1$}}%
              {\mbox{\boldmath$\scriptscriptstyle #1$}}}}

\newcommand{\Romannumcolor}[1]{\ensuremath{\textcolor{blue}{\Romannum{#1}}}}

\newcommand{\Fq}{\mathbb{F}_q}

\definecolor{darkred}{rgb}{0.7,0,0}

\newcommand{\smallsum}{{\textstyle\sum\nolimits}}

\renewcommand{\c}{\ve c}
\renewcommand{\r}{\ve r}

\newcommand{\e}{\ve e}

\newcommand{\Perr}{\mathrm{P}_{\mathrm{err}}}
\newcommand{\Prob}{\mathrm{P}}

\newcommand{\CRS}{\mathcal{C}_\mathrm{RS}}

\newcommand{\dH}{\mathrm{d}_\mathrm{H}}

\definecolor{ColorSecure}{rgb}{0,0,0} 
\definecolor{ColorNotSecure}{rgb}{0.7,0,0}
\newcommand{\LineTypeNotSecure}{dashed}

\newcommand{\wtH}{\mathrm{wt}_\mathrm{H}}
\newcommand{\Code}{\mathcal{C}}
\newcommand{\Ftwo}{\mathbb{F}_2}

\newcommand{\Acode}{\mathcal{A}}
\newcommand{\Bcode}{\mathcal{B}}

\newcommand{\COC}{\mathcal{C}_\mathrm{C}}

\newcommand{\Fqi}[1]{\mathbb{F}_{q^{#1}}}

\newcommand{\nB}{{n_b}}
\newcommand{\kB}{{k_b}}
\newcommand{\dB}{{d_b}}

\newcommand{\nA}{{n_a}}
\newcommand{\kA}{{k_a}}
\newcommand{\dA}{{d_a}}

\newcommand{\nOC}{{n_\mathrm{C}}}
\newcommand{\kOC}{{k_\mathrm{C}}}
\newcommand{\dOC}{{d_\mathrm{C}}}

\newcommand{\NN}{\mathbb{N}}

\newcommand{\eh}{\ve{h}} 
\newcommand{\ri}{\ve{r}} 
\newcommand{\rr}{\ve{r}'} 
\newcommand{\chat}{\hat{\ve{c}}}
\newcommand{\rhat}{\hat{\ve{r}}}

\newcommand{\frr}{\mathcal{RM}}

\newcommand{\Rstar}{R^\ast}

\newcommand{\gcandidate}{g}
\newcommand{\List}{\mathcal{L}}

\begin{document}

\title{Timing Attack Resilient Decoding Algorithms for Physical Unclonable Functions}
\author{\IEEEauthorblockN{Sven Puchinger$^{1}$, Sven Müelich$^{1}$, Antonia Wachter-Zeh$^{2}$ and Martin Bossert$^{1}$}
\IEEEauthorblockA{
$^1$Institute of Communications Engineering, Ulm University, Ulm, Germany\\
\texttt{\{sven.puchinger | sven.mueelich | martin.bossert\}@uni-ulm.de}\\
$^2$Institute for Communications Engineering, Technical University of Munich, Munich, Germany\\
\texttt{antonia.wachter-zeh@tum.de}
}
}

\maketitle

\begin{abstract}
This paper deals with the application of list decoding of Reed--Solomon codes to a concatenated code for key reproduction using Physical Unclonable Functions.
The resulting codes achieve a higher error-correction performance at the same code rate than known schemes in this scenario.
We also show that their decoding algorithms can be protected from side-channel attacks on the runtime both by masking techniques and by directly modifying the algorithms to have constant runtime.
\end{abstract}

\section{Introduction}

\noindent
A \emph{Physical Unclonable Function} (PUF) is a digital circuit that possesses an intrinsic randomness resulting from process variations.
This randomness is exploited to generate random keys for cryptographic applications.
An advantage of PUFs over other true random number generators is their ability to reproduce a key on demand.
Thus, no embedded physically secure non-volatile memory is needed.

However, the regeneration of a key is not perfect due to environmental factors such as temperature variations and aging effects of the digital circuit.
These variations can be seen as an erroneous channel and channel coding increases the reliability of key regenerations.
Error-correction methods for this purpose were considered in \cite{bosch2008efficient} (repetition, Reed-Muller (RM), Golay, BCH and concatenated codes), \cite{MaesCryptoPaper2012,MaesDiss2012} (concatenation of a repetition and BCH code), \cite{muelich2014error,puchinger2015error,hiller2015low} (generalized concatenated codes using RM and Reed--Solomon (RS) codes).

So far, most publications about error correction for PUFs have tried to find codes with low-complexity decoding methods (in time, area, etc.) and high decoding performance.
However, as for most other hardware security devices, PUFs need to be resistant against \emph{side-channel attacks}.
Their purpose is to obtain information about the secret by measurements, such as timing, energy consumption or electromagnetic fields. Throughout this paper, we only deal with side-channel attacks on the runtime, often called \emph{timing attacks}.

We consider RS codes in a concatenated coding scheme, where we use list decoding in order to increase the decoding radius beyond half the minimum distance.
In this way, smaller block error probabilities\footnote{In PUF literature, block error probability is often called \emph{failure rate}.} than the codes/decoders proposed in \cite{bosch2008efficient,MaesCryptoPaper2012,MaesDiss2012,muelich2014error,puchinger2015error,hiller2015low} can be achieved.

In addition, we protect the decoding algorithm from timing attacks.
We prove that the masking technique introduced in~\cite{merli2013protecting} is information-theoretically secure and propose methods for preventing attacks on decoders with unmasked inputs.

Section~\ref{sec:preliminaries} deals with preliminaries.
We propose to use list decoding of RS codes in error correction for PUFs in Section~\ref{sec:construction} and analyze their performance.
Sections~\ref{sec:preventing}, \ref{sec:Masking}, and \ref{sec:Classical} present ideas of preventing timing-attacks on the list decoding algorithm and Section~\ref{sec:conclusion} concludes the paper.

\section{Prelimiaries}\label{sec:preliminaries}

\noindent
In this paper, $\Code=\Code(q; n,k,d)$ is a linear code over a finite field $\Fq$ ($q$ prime power) of length $n$, dimension $k$ and minimum distance $d$. If the field is clear from the context, we write $\Code(n,k,d)$.
We use the classical Shannon entropy
\begin{align*}
H(X) = -\textstyle\sum_{x}f_X(x) \log_2(f_X(x)), 
\end{align*}
where the input $X$ is considered to be a random variable. E.g., if $\c$ is a codeword that is drawn uniformly at random from a code $\Code(n,k,d)$, its entropy $H(\c)$ is $k$.

\subsection{Error-Correction in PUF-based Key Reproduction}

We briefly describe key reproduction using PUFs with the \emph{code-offset} method\footnote{We consider only linear codes in this paper while the code-offset method generally also works for non-linear codes.} \cite{Dodis2008}, as illustrated in Figure~\ref{fig:puf}. A comprehensive overview of PUFs and how to use error-correction for key reproduction can be found in \cite{MaesDiss2012,boehm2012puf,wachsmann2015puf}.

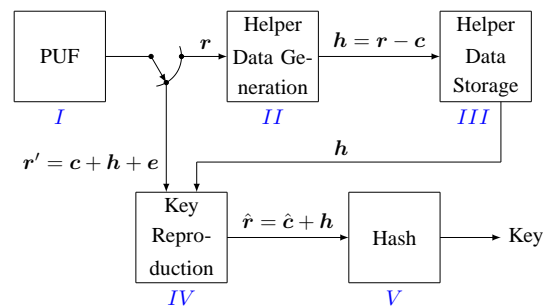
\begin{figure}[htb]
\begin{center}
{
\resizebox{0.4\textwidth}{!}{
\centering
\begin{tikzpicture}[scale=0.5]

\draw (0,0) -- (0,3);
\draw (0,3) -- (3,3);
\draw (3,3) -- (3,0);
\draw (3,0) -- (0,0);
\draw (1.5,1.5) node {PUF};
\draw (1.5,0) node [below] {$\Romannumcolor{1}$};

\draw (3,1.5) -- (4.5,1.5);
\draw[->,>=latex] (5,1.5-0.86602540378) -- (5,-3);
\draw[fill] (4.5,1.5) circle (2pt);
\draw[fill] (5.5,1.5) circle (2pt);
\draw[fill] (5,1.5-0.86602540378) circle (2pt);
\draw (5,-2) node [left] {$\rr = \c + \eh + \e$};
\draw [domain=-80:20] plot ({4.5+cos(\x)}, {1.5+sin(\x)});
\draw[->,>=latex] (4.5,1.5) -- (5,1.5-0.86602540378);

\draw (7,0) -- (7,3);
\draw (7,3) -- (10,3);
\draw (10,3) -- (10,0);
\draw (10,0) -- (7,0);
\draw (8.5,2.5) node {Helper};
\draw (8.5,1.5) node {Data Ge-};
\draw (8.5,0.5) node {neration};
\draw (8.5,0) node [below] {$\Romannumcolor{2}$};

\draw[->,>=latex] (5.5,1.5) -- (7,1.5);
\draw[->,>=latex] (10,1.5) -- (14,1.5);
\draw (6.3,1.5) node [above] {$\ri$};
\draw (12,1.5) node [above] {$\eh = \ri-\c$};

\draw (14,0) -- (14,3);
\draw (14,3) -- (17,3);
\draw (17,3) -- (17,0);
\draw (17,0) -- (14,0);
\draw (15.5,2.5) node {Helper};
\draw (15.5,1.5) node {Data};
\draw (15.5,0.5) node {Storage};
\draw (15.1,0) node [below] {$\Romannumcolor{3}$};

\draw (16,0) -- (16,-2);
\draw (16,-2) -- (6,-2);
\draw[->,>=latex] (6,-2) -- (6,-3);
\draw (10.75,-2) node [above] {$\eh$};

\draw (4,-3) -- (7,-3);
\draw (7,-3) -- (7,-6);
\draw (7,-6) -- (4,-6);
\draw (4,-6) -- (4,-3);
\draw (5.5,-3.5) node {Key};
\draw (5.5,-4.5) node {Repro-};
\draw (5.5,-5.5) node {duction};
\draw (5.5,-6) node [below] {$\Romannumcolor{4}$};

\draw[->,>=latex] (7,-4.5) -- (11,-4.5);
\draw (9,-4.5) node [above] {$\rhat = \chat + \eh$};

\draw (11,-3) -- (14,-3);
\draw (14,-3) -- (14,-6);
\draw (14,-6) -- (11,-6);
\draw (11,-6) -- (11,-3);
\draw (12.5,-4.5) node {Hash};
\draw (12.5,-6) node [below] {$\Romannumcolor{5}$};

\draw[->,>=latex] (14,-4.5) -- (16,-4.5);
\draw (16,-4.5) node [right] {Key};

\end{tikzpicture}
}
}
\end{center}
\caption{Key Generation and Reproduction based on PUFs \cite{puchinger2015error}.} 
\label{fig:puf}
\end{figure}

An initial \emph{response} $\ri \in \Ftwo^n$ with entropy $H(\ri) \approx n$ is generated by the \emph{PUF} ($\Romannumcolor{1}$) and a random codeword $\c \in \Code$ of a binary linear code $\Code(n,k,d)$ is subtracted from $\ri$ in the \emph{Helper Data Generation} ($\Romannumcolor{2}$). The resulting \emph{helper data} $\eh$ is then stored in the \emph{Helper Data Storage} ($\Romannumcolor{3}$) and can be made publicly available since knowing $\eh$ leaves the attacker with an uncertainty of the choice of the codeword. More precisely, for a uniformly drawn codeword, we obtain
\begin{align*}
H(\ri \, | \, \eh) &= H(\ri , \c) - H(\eh) = H(\ri) + H(\c) - H(\eh) \\
&\geq H(\ri) - (n-k) \approx k.
\end{align*}

In the \emph{reproduction phase}, the PUF outputs a response $\r'~\in~\Ftwo^n$, which differs from $\ri$ by an error $\e$ whose physical causes are environmental conditions such as temperature and aging, and we can write $\rr = \ri + \e$,
where $\e$ is often\footnote{In practice, each PUF bit exhibits a unique individual bit error rate due to the imperfect behavior of the digital circuit. Some papers therefore consider different channel models, cf.~\cite{maes2013accurate,delvaux2015helper}.} modelled as a \emph{binary symmetric channel} (BSC) with crossover probability $p$ (e.g., $p=0.14$ in \cite{MaesCryptoPaper2012}).

In order to reproduce the original sequence $\ri$, the \emph{Key Reproduction} unit ($\Romannumcolor{4}$) subtracts the helper data from $\rr$ and decodes the resulting word
\begin{align*}
\rr-\eh = \c + \eh + \e - \eh = \c + \e
\end{align*}
using a decoder of $\Code$ and obtains a codeword $\chat$. If the number of errors $\wtH(\e)$ is within the error-correction capability of the decoder, $\chat = \c$ and we can compute the original sequence as $\ri = \chat+\eh$.
The result is then usually hashed ($\Romannumcolor{5}$) in order to obtain keys of length $m \leq k$, which ideally are uniformly distributed over $\Ftwo^{m}$.

\subsection{Reed--Solomon Codes and List Decoding}
\label{subsec:ListDecoding}

\emph{Reed--Solomon} (RS) codes are algebraic codes with a variety of applications, the largest possible minimum distance, and
efficient decoding algorithms, both for decoding up to and beyond half the minimum distance.
Let $q$ be a prime power and $\Fq$ be the finite field of size $q$ and let $\Fq[x]$ denote the polynomial ring over $\Fq$.
\begin{definition}
Let $\alpha_1,\dots,\alpha_n \in \Fq$ be distinct.
An $(n,k)$ RS code of dimension $k<n$ and length $n$ is given by the set
\begin{align*}
\CRS = \left\{ (f(\alpha_1), \dots, f(\alpha_n)): f \in \Fq[x], \deg f < k \right\} \subseteq \Fq^n.
\end{align*}
\end{definition}
It can be shown that the minimum distance of an $(n,k)$ RS code is $d = n-k+1$.
There are several algorithms for uniquely decoding up to $\left\lfloor(d-1)/2\right\rfloor$ errors with RS codes, see e.g., \cite{bossert1999channel}.
\emph{List decoding} generalizes this concept for $\tau > \left\lfloor(d-1)/2\right\rfloor$ errors.
A \emph{list decoder} guarantees to return a list of all codewords $\c$ that fulfill $\dH(\c,\r) \leq \tau$
for a given decoding radius $\tau$ and the received word $\r$.
For RS codes, the Guruswami--Sudan decoding algorithm~\cite{guruswami1999improved} 
accomplishes list decoding in polynomial time for any $\tau <  n-\sqrt{n(k-1)}$.
The algorithm is based on the following interpolation problem.
\begin{problem}\label{prob:GS}
Given $\r = (r_1,\dots,r_n)\in \Fq^n$, find a non-zero bivariate polynomial $Q(x,y) \in \Fq[x,y]$ of the form $Q(x,y) = \sum_{j=0}^{\ell}{Q_j(x) y^j}$, such that for given integers $s$, $\tau$ and $\ell$:
\begin{enumerate}
	\item $(\alpha_i,r_i)$ are zeros of $Q(x,y)$ of multiplicity $s$,\\ $\forall i=1,\dots,n,$
	\item $\deg Q_j(x) \leq s(n-\tau)-1-j(k-1)$, $\forall j=0,\dots,\ell.$
\end{enumerate}
\end{problem}
The multiplicity $s$ can always be chosen large enough such that any $\tau < n-\sqrt{n(k-1)}$ can be achieved.
The Guruswami--Sudan algorithm returns a list of all polynomials that satisfy $(y-f(x))|Q(x,y)$.
It was proven in \cite{guruswami1999improved} that this list of  polynomials includes all evaluation polynomials $f(x)$, which generate codewords with $\dH(\c,\r) \leq \tau$.
The size of this list is bounded by a polynomial function in the code-length whenever $\tau<n-\sqrt{n(k-1)}$.
The algorithm consists of two main steps: the interpolation step and the root-finding step.
There are several efficient implementations, for both, the interpolation step~\cite{Alekhnovich2005Linear,ZehGentnerAugot-FIAGuruswamiSudan_2011} and the root-finding step~\cite{RothRuckenstein2000}.
Also, efficient VLSI implementations exist, e.g.~\cite{gross2002vlsi}.

\subsection{Reed--Muller Codes}
\label{subsec_rm}

A \emph{Reed--Muller} (RM) code $\frr(r,m)$ of order $r$ with $r \leq m$ is a binary linear code with parameters $n = 2^{m}$, $k = \sum_{i = 0}^{r} \binom{m}{i}$ and $d=2^{m-r}$.
It can be defined recursively using the \emph{Plotkin Construction} \cite{bossert1999channel}:
\begin{align*}
\frr(r,m) &:= \left\{ (\ve{a} | \ve{a} + \ve{b} ) : \begin{array}{l} \ve{a} \in \frr(r,m-1) \\ \ve{b} \in \frr(r-1,m-1) \end{array}\right\}
\end{align*}
with $\frr(0,m) := \Code(2^m,1,2^m)$ (\emph{Repetition code}) and $\frr(m-1,m) := \Code(2^m,2^m-1,2)$ (\emph{Parity Check code}) for all $m$.
RM codes have been proposed for PUF key reproduction in \cite{maes2009low,Hiller2012,muelich2014error,puchinger2015error}, and an efficient implementation of the decoding algorithm in FPGAs was presented in \cite{hiller2015low}.

\subsection{Concatenated Codes}

Concatenation \cite{forney1966concatenated} of two linear codes is a technique for generating new codes from existing ones, while keeping encoding and decoding complexities small.

We describe code concatenation as in \cite{sendrier1998concatenated}.
Let $\Bcode(q;\nB,\kB,\dB)$ (\emph{inner code}) and $\Acode(q^\kB;\nA,\kA,\dA)$ (\emph{outer code}) be two linear codes for a suitable choice of $q,\nB,\kB,\dB,\nA,\kA$ and $\dA$.
We use an encoding mapping for the code $\Acode$, i.e., an $\Fq$-linear map $\theta \, : \, \Fqi{\kB} \, \to \, \Bcode$. We can extend the mapping to matrices by applying it row-wise:
\begin{eqnarray*}
\Theta \, : \, \Fqi{\kB}^\nA \, &\to& \, \Bcode^\nA \\
\begin{bmatrix}a_1 \\ \vdots \\ a_\nA \end{bmatrix} \, &\mapsto& \, \begin{bmatrix}\theta(a_1) \\ \vdots \\ \theta(a_\nA) \end{bmatrix}.
\end{eqnarray*}

\begin{definition}[Concatenated Code]
Let $\Acode$, $\Bcode$, $\nA$ and $\Theta$ be as above. The corresponding \emph{concatenated code} is given as
\begin{eqnarray*}
\COC = \Theta(\Acode) \subseteq \Bcode^{\nA}
\end{eqnarray*}
\end{definition}

We call the set of positions containing the $i$th inner codeword $\theta(a_i)$ the $i$th \emph{inner block}. Codewords are often represented as matrices, where the $i$th row contains the $i$th block.
Due to its construction, a concatenated code is $\Fq$-linear.
The code has $(q^\kB)^\kA = q^{\kB \cdot \kA}$ codewords, each of it consisting of $\nA$ many codewords from $\Bcode$, resulting in a code-length of $\nA \cdot \nB$ elements of $\Fq$.
Thus, the code has parameters
\begin{eqnarray*}
\COC(q;\nOC= \nA \cdot \nB, \kOC = \kB \cdot \kA, \dOC),
\end{eqnarray*}
where $\dOC \geq \dB \cdot \dA$ is the minimum distance.
Although $\dOC$ might be small, often a lot more errors than half-the-minimum distance can be corrected.
Concatenation of codes and standard decoders have been suggested for the PUF scenario in~\cite{bosch2008efficient, MaesCryptoPaper2012}.

The construction can be extended to \emph{generalized concatenated codes} \cite{blokh1974coding}, see also~\cite{bossert1999channel}.
Generalized concatenated codes were proposed for error correction in key reproduction using PUFs in \cite{muelich2014error,puchinger2015error} and a low-complexity decoding design for FPGAs was presented in \cite{hiller2015low}.

\section{Code Constructions and List Decoding in the PUF Scenario}
\label{sec:construction}

Choosing codes and decoders for error correction in key reproduction using PUFs is subject to many constraints that arise from their physical properties.
Typical design criteria \cite{MaesCryptoPaper2012,puchinger2015error} are listed below.
\begin{itemize}
\item Choose a dimension that fulfills $H(\mathrm{key}) \leq H(\ri)-n+k$, where $H(\mathrm{key})$ is the desired entropy of the extracted key.
\item Minimize the code-length $n$.
\item Obtain a \emph{block error probability} $\Perr$ that is below a certain threshold (e.g., $10^{-9}$).
\item Find efficient decoders (in time, area, memory, etc.).
\item Provide resistance to side-channel attacks (with respect to time, energy, electro-magnetic radiation, etc.).
\end{itemize}
In the following, we recall one of the code constructions in~\cite{puchinger2015error} and show that by using list decoding we can improve the error-correction performance of this scheme.

\subsection{Code Construction}

As in \cite{puchinger2015error}, we choose the inner code to be a binary Reed--Muller code $\Bcode(2;\nB,\kB,\dB) = \frr(r,m)$ (cf.~Section~\ref{subsec_rm}) and a Reed--Solomon code $\Acode(2^\kB;\nA,\kA,\nA-\kA+1) = \CRS(\nA,\kA)$ (cf.~Section~\ref{subsec:ListDecoding}) as outer code.

\subsection{Decoding}

Decoding works in two steps.
First, the inner blocks of the received word are decoded using the inner RM code $\Bcode$.
The respective decoding result either corresponds to an element in $\Fqi{\kB}$ or to an erasure.
Afterwards, the vector containing the decoding results of the inner blocks is decoded in the RS code.

If a list decoder (cf.~Section~\ref{subsec:ListDecoding}) is used in this step, more errors can be corrected than with power decoding, which was proposed in \cite{puchinger2015error}.
The following example compares our coding scheme with known ones for the scenario considered in \cite{bosch2008efficient,MaesCryptoPaper2012,MaesDiss2012,muelich2014error,puchinger2015error} (BSC with $p=0.14$, $H(\r)\geq 128$, goal $\Perr<10^{-9}$).

\begin{example}
We consider the construction in \cite[Section~IV-C]{puchinger2015error}, namely an inner RM code with parameters $\Bcode(2;32,6,16) = \frr(1,5)$ and an outer RS code $\Acode(2^6;64,22,43)$. The resulting concatenated code has parameters $\COC(2; 2048, 132, \geq 688)$.
Using the algorithms proposed in \cite{puchinger2015error}, the resulting block error probability is $\Perr \approx 6.79 \cdot 10^{-37}$.

Maximum likelihood decoding of the inner RM code transforms the channel into a binary error and erasure channel with $\Prob(\text{error}) = 0.003170$ and  $\Prob(\text{erasure}) = 0.017605$ \cite{puchinger2015error}.
Since the minimum distance of the RS codes is $d=43$, unique decoding is possible up to $21$ errors and list decoding with the Guruswami--Sudan algorithm up to $\lceil n-\sqrt{n(k-1)}\rceil-1 = 27$ errors. When erasures are present, the Guruswami--Sudan decoder simply considers only non-erased positions in the interpolation step.
Let $t$ and $\varepsilon$ denote the number of errors and erasures, respectively.
Then, the block error probability is
\begin{align*}
\Perr &= \sum\limits_{i=0}^{n} \Prob(\varepsilon=i) \Prob(t\geq n-i-\sqrt{(n-i)(k-1)}) \\
&\approx 3.5308 \cdot 10^{-46},
\end{align*}
which is significantly smaller than for unique decoding, cf.~\cite{puchinger2015error}.

When replacing the outer code by the RS code $\CRS(2^6;34,22,13)$, the concatenated code has parameters $\COC(2; 1088, 132, \geq 208)$ and, using list decoding, the block error probability is
\begin{align*}
\Perr \approx 1.9981 \cdot 10^{-10} < 10^{-9},
\end{align*}
which is approximately the same as in \cite{puchinger2015error} while reducing the length of the concatenated code from $1152$ to $1088$.

Using generalized concatenated codes in combination with list decoding, also the block error probability of the other code constructions in \cite{puchinger2015error} can be decreased. Since the error correction schemes in \cite{puchinger2015error} decreased the block error probabilities and code-lengths simultaneously
compared to the constructions in \cite{bosch2008efficient,MaesCryptoPaper2012,MaesDiss2012,muelich2014error}, our results also improve upon them.
\end{example}

\subsection{Optimal Rates in the PUFKY~\cite{MaesCryptoPaper2012} scenario}

In general, we would like to know how close to an optimal solution our error correction schemes are. When comparing it to the capacity of the binary symmetric channel,
\begin{align*}
C = 1 -h(p) = 1 + p\log_2(p) + (1-p)\log_2(1-p),
\end{align*}
one will notice that the rates of most of the existing schemes are far away from this upper bound, which is expectable for finite block lengths.
It was proven in \cite[Theorem~52]{polyanskiy2010channel} that the maximal achievable rate $\Rstar(n,p,\Perr)$ of a code of length $n$ whose codewords are transmitted through a BSC with crossover probability $0<p<\tfrac{1}{2}$ with maximal block error probability $\Perr$ is
\begin{align*}
\Rstar(n,p,\Perr) = C-\sqrt{\tfrac{V}{n}} Q^{-1}(\Perr) + \tfrac{\log_2 n}{2n} + O(\tfrac{1}{n}),
\end{align*}
where
\begin{align*}
V = p (1-p) \log_2^2\left( \tfrac{1-p}{p} \right), \quad
Q(x) = {\int\nolimits}_{x}^{\infty} \tfrac{1}{\sqrt{2 \pi}} \mathrm{e}^{-\tfrac{y^2}{2}} dy.
\end{align*}

In \cite{MaesCryptoPaper2012,muelich2014error,puchinger2015error}, a BSC with crossover probability $p=0.14$ was considered and $\Perr < 10^{-9}$ was demanded. In this case, the capacity of the BSC is $C \approx 0.5842$, but the actual maximal achievable rates $\Rstar$ are much smaller. Table~\ref{tab:summary} shows how close the rates of existing code constructions and of our new construction are to the optimal rates.
\begin{table}[h]
\caption{Comparison between code constructions and decoders for PUFs based on concatenated codes.}
\label{tab:summary}

{
\renewcommand{\arraystretch}{1.5}
\centering
\begin{tabular}{p{1.6cm}|p{1.4cm}|p{0.4cm}|p{0.5cm}|p{0.7cm}|p{0.7cm}|p{0.7cm}}
\small
$\Acode$/$\Bcode$ (ref.)
& $\Perr$
& $k$ 
& $n$
& $R$
& $\Rstar$
& $R/\Rstar$ \\
\hline \hline
BCH/Rep. \cite{MaesCryptoPaper2012}
& $1.0 \cdot10^{-9}$
& $174$
& $2226$
& $0.0782$
& $0.3027$
& $0.2582$ \\
\hline
RS/RM$^\mathrm{u}$ \cite{puchinger2015error}
& $1.2 \cdot 10^{-10}$
& $132$
& $1152$
& $0.1146$
& $0.2506$
& $0.4573$ \\
\hline
RS/RM$^{\mathrm{\ell},\mathrm{t}}$
& $2.0 \cdot 10^{-10}$
& $132$
& $1088$
& $0.1213$
& $0.2481$
& $0.4890$ \\
\end{tabular}
}

\vspace{2ex}

\emph{Legend:} Outer code $\Bcode$, inner code $\Acode$, block error probability $\Perr$, rate $R=k/n$, maximal possible rate $\Rstar(n,0.14,\Perr)$, ratio to optimality $R/\Rstar$.
$^\mathrm{u}$Decoder based on unique decoding of the RS code (cf.~\cite{puchinger2015error}).
$^\mathrm{\ell}$Decoder based on list decoding of the RS code (cf.~Section~\ref{subsec:ListDecoding}).
$^{\mathrm{t}}$This paper.
\end{table}

We conclude that using list decoding, the error-correction capabilities of (generalized) concatenated code constructions based on outer RS codes can be improved significantly. Also, the coding schemes achieve approximately half of the maximum possible rates in the scenario considered in \cite{MaesCryptoPaper2012}, which is a large value for a practical coding scheme.
However, this gain comes at the cost of increased time and space complexity and therefore a larger power and area consumption.

\section{Preventing Timing Attacks}
\label{sec:preventing}

This section deals with securing the decoding algorithms of the code constructions considered in this paper against side-channel attacks on the runtime.
A side-channel attacker tries to obtain information from the hardware implementation of the PUF, which includes
runtime, power consumption, and electromagnetic radiation.
For example, ring oscillator PUFs compare the frequencies of two ring oscillators and therefore inevitably induce
an electro-magnetic emission depending on their frequencies
that leads to side information~\cite{merli2011side}.
The paper~\cite{merli2011side} deals with side-channel attacks on the helper data.
In~\cite{merli2013protecting}, it was proposed to add another random codeword (called \emph{codeword masking}) on the helper data before the key reproduction. 
In~\cite{dai2009study}, it was analyzed how much information is leaked from the power consumption when storing a codeword of a single-parity check code in a static memory.

However, to the best of our knowledge, there are no publications that focus on attacking the \emph{decoding process} itself, e.g., the runtime and power consumption while executing the decoding algorithm.
It is therefore important that a decoder has constant runtime and constant power consumption, independent of the received word.
In the following, we design a list decoder with constant runtime.\footnote{Unlike most publications in the field of side-channel attacks, we do not provide an FPGA implementation, but analyze our algorithm theoretically.
Measuring the side-channel attack resilience of such an implementation is a necessary step for further research.}

We focus on side-channel attacks of the decoding algorithm. Therefore, assume that only parts of the key reconstruction functions are attackable, as illustrated in Figure~\ref{fig:attacker_model}.

\begin{figure}[h]
\begin{center}
\resizebox{0.45\textwidth}{!}{
\begin{tikzpicture}[scale=0.5]

\draw[thick, ColorSecure] (0,0) -- (0,3);
\draw[thick, ColorSecure] (0,3) -- (3,3);
\draw[thick, ColorSecure] (3,3) -- (3,0);
\draw[thick, ColorSecure] (3,0) -- (0,0);
\draw (1.5,1.5) node {PUF};
\draw (1.5,0) node [below] {$\Romannumcolor{1}$};

\draw[thick, ColorSecure] (3,1.5) -- (4.5,1.5);
\draw[thick, ColorSecure,->,>=latex] (5,1.5-0.86602540378) -- (5,-3);
\draw[ColorSecure, fill] (4.5,1.5) circle (2pt);
\draw[ColorSecure, fill] (5.5,1.5) circle (2pt);
\draw[ColorSecure, fill] (5,1.5-0.86602540378) circle (2pt);
\draw (5,-2) node [left] {$\rr = \c + \eh + \e$};
\draw[thick, ColorSecure] [domain=-80:20] plot ({4.5+cos(\x)}, {1.5+sin(\x)});
\draw[thick, ColorSecure, ->,>=latex] (4.5,1.5) -- (5,1.5-0.86602540378);

\draw[thick, ColorSecure] (9,0) -- (9,3);
\draw[thick, ColorSecure] (9,3) -- (12,3);
\draw[thick, ColorSecure] (12,3) -- (12,0);
\draw[thick, ColorSecure] (12,0) -- (9,0);
\draw (10.5,2.5) node {Helper};
\draw (10.5,1.5) node {Data Ge-};
\draw (10.5,0.5) node {neration};
\draw (10.5,0) node [below] {$\Romannumcolor{2}$};

\draw[thick, ColorSecure, ->,>=latex] (5.5,1.5) -- (9,1.5);
\draw (7.3,1.5) node [above] {$\ri = \c + \eh$};

\draw[thick, \LineTypeNotSecure, ColorNotSecure, ->,>=latex] (12,1.5) -- (14,1.5);
\draw (13,1.5) node [above] {$\eh$};

\draw[thick, \LineTypeNotSecure, ColorNotSecure] (14,0) -- (14,3);
\draw[thick, \LineTypeNotSecure, ColorNotSecure] (14,3) -- (17,3);
\draw[thick, \LineTypeNotSecure, ColorNotSecure] (17,3) -- (17,0);
\draw[thick, \LineTypeNotSecure, ColorNotSecure] (17,0) -- (14,0);
\draw (15.5,2.5) node {Helper};
\draw (15.5,1.5) node {Data};
\draw (15.5,0.5) node {Storage};
\draw (15.1,0) node [below] {$\Romannumcolor{3}$};

\draw[thick, \LineTypeNotSecure, ColorNotSecure] (16,0) -- (16,-2);
\draw[thick, \LineTypeNotSecure, ColorNotSecure] (19,-2) -- (6,-2);
\draw[thick, \LineTypeNotSecure, ColorNotSecure, ->,>=latex] (6,-2) -- (6,-3);
\draw[thick, \LineTypeNotSecure, ColorNotSecure, ->,>=latex] (19,-2) -- (19,-3);
\draw (10.75,-2) node [above] {$\eh$};

\draw[thick, ColorSecure] (4,-3) -- (7,-3);
\draw[thick, ColorSecure] (7,-3) -- (7,-6);
\draw[thick, ColorSecure] (7,-6) -- (4,-6);
\draw[thick, ColorSecure] (4,-6) -- (4,-3);
\draw (5.5,-3.5) node {Pre-};
\draw (5.5,-4.5) node {Proces-};
\draw (5.5,-5.5) node {sing};
\draw (6,-6) node [below] {$\Romannumcolor{4}$};

\draw[thick, ColorSecure, ->,>=latex] (7,-4.5) -- (11,-4.5);
\draw (9,-4.5) node [above] {$\varphi(\c+\e)$};

\draw[thick, \LineTypeNotSecure, ColorNotSecure] (11,-3) -- (14,-3);
\draw[thick, \LineTypeNotSecure, ColorNotSecure] (14,-3) -- (14,-6);
\draw[thick, \LineTypeNotSecure, ColorNotSecure] (14,-6) -- (11,-6);
\draw[thick, \LineTypeNotSecure, ColorNotSecure] (11,-6) -- (11,-3);
\draw (12.5,-4.5) node {Decoding};
\draw (12.5,-6) node [below] {$\Romannumcolor{5}$};

\draw[thick, ColorSecure, ->,>=latex] (14,-4.5) -- (18,-4.5);
\draw (16,-4.5) node [above] {$\varphi(\chat)$};

\draw[thick, ColorSecure] (18,-3) -- (21,-3);
\draw[thick, ColorSecure] (21,-3) -- (21,-6);
\draw[thick, ColorSecure] (21,-6) -- (18,-6);
\draw[thick, ColorSecure] (18,-6) -- (18,-3);
\draw (19.5,-3.5) node {Post-};
\draw (19.5,-4.5) node {Proces-};
\draw (19.5,-5.5) node {sing};
\draw (20,-6) node [below] {$\Romannumcolor{6}$};

\draw[thick, ColorSecure] (5,-6) -- (5,-8);
\draw[thick, ColorSecure] (5,-8) -- (19,-8);
\draw[thick, ColorSecure,->, >=latex] (19,-8)-- (19,-6);
\draw (12.5,-8) node [below] {$\varphi^{-1}$}; 

\draw[thick, ColorSecure, ->,>=latex] (21,-4.5) -- (23,-4.5);
\draw (23,-4.5) node [right] {Key};

\draw (18.5,2.5) node [right] {Legend:};
\draw[thick, ColorSecure] (19,1.5) -- (20,1.5);
\draw (20,1.5) node [right] {not attackable};
\draw[thick, \LineTypeNotSecure, ColorNotSecure] (19,0.5) -- (20,0.5);
\draw (20,0.5) node [right] {attackable};

\end{tikzpicture}
}
\end{center}
\caption{Attacker Model.} 
\label{fig:attacker_model}
\end{figure}
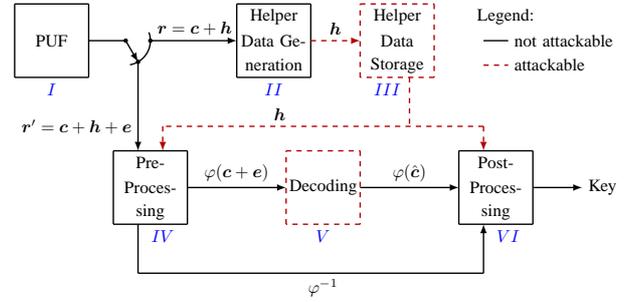
{
The PUF ($\Romannumcolor{1}$) and the Helper Data Generation ($\Romannumcolor{2}$) are assumed to be secure here.
As mentioned before, the Helper Data Storage ($\Romannumcolor{3}$) can be read by an attacker without obtaining more information about $\ri$ than contained in the random choice of $\c$.
Compared to the model in Figure~\ref{fig:puf}, the Key Reproduction unit is subdivided into \emph{Preprocessing} ($\Romannumcolor{4}$), the \emph{Decoder} ($\Romannumcolor{5}$), and the \emph{Post-Processing} unit ($\Romannumcolor{6}$). The latter also includes the hashing of the key here.
Preprocessing ($\Romannumcolor{4}$) is assumed to be not attackable. We distinguish two types of preprocessing:

\begin{enumerate}
\item \emph{Classical preprocessing}: Compute $\c + \e = \rr- \eh$ and hand it over to the Decoding ($\Romannumcolor{5}$) unit.
\item \emph{Masking}: Choose random function $\varphi$ such that the decoder can map $\varphi(\c + \e)$ into $\varphi(\c)$, but even if $\varphi(\c + \e)$ can be obtained by an attacker, the uncertainty about $\ri$ is not decreased.
\end{enumerate}
The Decoder ($\Romannumcolor{5}$) can be attacked.
In the Post-Processing unit ($\Romannumcolor{6}$), we compute $\varphi^{-1}(\varphi(\chat))+\eh = \chat+\eh$.
If decoding was successful, we obtain the original response $\chat + \eh = \c + \eh = \ri$.
The key is then often computed as a hash of $\ri$ \cite{MaesCryptoPaper2012}.

\section{Attack Resistance by Masking}
\label{sec:Masking}

\subsection{Codeword Masking}

One method to hide the actual codeword $\c$ from an attacker who can retrieve the processed data $\varphi(\c + \e)$ is the \emph{codeword masking} technique proposed in~\cite{merli2013protecting}, where a random codeword $\c'$ is chosen and added to $\c + \e$, i.e.,
\begin{align*}
\varphi(\c + \e) = \c'+\c+\e.
\end{align*}
The technique is based on general masking schemes for preventing DPAs.
In~\cite{merli2013protecting}, no proof was given that the method actually masks well.
The following theorem proves that even if an attacker is able to retrieve both the helper data~$\eh$ and the masked word $\c'+\c+\e$, the remaining uncertainty is still large enough.

\begin{theorem}
$H(\ri \, | \, (\c'+\c+\e,\eh)) \geq H(\ri) - (n-k)$
\end{theorem}

\begin{proof}
We know that $\ri,\c,\c',\e$ are pairwise independent. Also, $\c$ and $\c'$ are uniformly drawn from the code, so
\begin{align}
H(\c+\e) = H(\c+\c'+\e). \label{eq:Hce=Hcc'e}
\end{align}
In general, it holds that
\begin{align}
H(\c+\c'+\e,\eh) \leq H(\c+\c'+\e)+H(\eh). \label{eq:H_subadditivity}
\end{align}
Since we can compute $(\ri, \c'+\c+\e, \c)$ from $(\ri, \c'+\c+\e, \eh)$ and vice versa, we have
\begin{align}
&H(\ri, \c'+\c+\e, \eh) = H(\ri, \c'+\c+\e, \c) \notag \\
&= H(\ri \mid (\c'+\c+\e,\c)) + H(\c'+\c+\e, \c) \notag \\
&= H(\ri) + H(\c'+\c+\e, \c) \notag \\
&= H(\ri) + H(\c'+\c+\e \mid \c) + H(\c) \notag \\
&= H(\ri) + H(\c'+\e) + H(\c) \label{eq:Hrcc'eh}
\end{align}
Hence, we obtain
\begin{align*}
&H(\ri \mid (\c'+\c+\e,\eh)) \\
&= H(\ri, \c'+\c+\e, \eh) - H(\c'+\c+\e,\eh) \\
&\overset{\eqref{eq:H_subadditivity},\eqref{eq:Hrcc'eh}}{\geq} H(\ri) + H(\c'+\e) + H(\c) - H(\c+\c'+\e)-H(\eh) \\
&\overset{\eqref{eq:Hce=Hcc'e}}{=} H(\ri) + H(\c) - H(\eh) 
= H(\ri) + k - H(\eh) \\
&\geq H(\ri) - (n-k). \qedhere
\end{align*}
\end{proof}
Note that if $H(\ri) = n$, then $H(\ri \, | \, (\c'+\c+\e,\eh)) \geq k$.

\subsection{Alternative Masking Techniques}

Other than adding a codeword to the processed word, the only masking operations that do not change the Hamming weight of the error (i.e., the hardness of the decoding problem) are the Hamming-metric isometries. Over $\Ftwo^n$, those are exactly all permutations of positions since the other possibility, the Frobenius automorphism $\cdot^2$, is the identity map in $\Ftwo$.

In the case of RS codes, the decoder can handle a permutation $\pi$ of positions since $\pi(\c)$ is also a codeword of an RS code with permuted code locators $\alpha_i$.
Thus, $\pi(\c+\e) = \pi(\c) + \pi(\e)$ with $\wtH(\pi(\e)) = \wtH(\e)$, we can obtain $\pi(\c)$ from $\pi(\c+\e)$ using a decoder for Reed--Solomon codes whenever it is possible to correct $\e$ in $\c+\e$.

Note that if $\pi$ is not an element of the automorphism group of the code, the decoder must know the permutation $\pi$.
If it is in the automorphism group, then $\pi(\c)-\c$ is a codeword and the method is equivalent to codeword masking.

\section{Attack Resistance by Constant Runtime (Classical Preprocessing)}
\label{sec:Classical}

\subsection{Realizing Finite Field Operations}
\label{subsec:fieldop}

The codes used in Section~\ref{sec:construction} can be decoded using algebraic decoding algorithms that perform operations in finite fields.
For error-correction in key regeneration using PUFs, we usually consider fields of characteristic $2$, i.e., $\mathbb{F}_{2^m}$ for some $m \in \NN$.
Motivated by elliptic-curve cryptography, operations in these fields have recently been made resistant against timing-attacks while preserving sufficient speed in \cite{pamula2015fast}.
For small fields, lookup tables could be used. E.g., the field~$\mathbb{F}_{2^6}$ used in the construction in Section~\ref{sec:construction} would require tables of $2^{6 \cdot 2} = 4096$ entries.

Based on these considerations, we assume that field operations in $\mathbb{F}_{2^m}$, also if a zero is involved, are constant in runtime.

\subsection{Outer Code: List Decoding of RS Codes}
\label{subsec:outer_code}

\subsubsection{Interpolation step}

The interpolation step consists of finding a bivariate polynomial
\begin{align*}
Q(x,y) = \smallsum_{\eta=0}^{\ell} Q_\eta(x) y^\eta = \smallsum_{\eta=0}^{\ell} \smallsum_{\mu=0}^{d_\eta} Q_{\eta,\mu} x^\mu y^\eta,
\end{align*}
where $d_\eta = s(n-\tau)-1-\eta(k-1)$, satisfying properties 1)-2) of Problem~\ref{prob:GS}. This corresponds to finding a non-zero solution $Q_{\eta,\mu} \in \Fq$ for $0 \leq \mu \leq d_\eta$ and $0\leq \eta \leq \ell$ of the system
\begin{align*}
\smallsum_{\eta=0}^{\ell} \smallsum_{\mu=0}^{d_\eta} \tbinom{\eta}{h} \tbinom{\mu}{j} Q_{\eta,\mu} \alpha_i^{\mu-j} r_i^{\eta-h}=0
\end{align*}
with $i=0,\dots,n$ and $h+j<s$.

There are many efficient algorithms for finding such a solution which are asymptotically faster than simply solving this system without considering its structure. However, these fast methods might reveal side-information about the processed data since their runtime depends on the received word $\r$.

Therefore, we propose to solve the system using ``naive'' Gaussian elimination where we always apply a row operation, even when an element is already zero (simply add a zero row to it).
The resulting algorithm always performs the same number of field additions and multiplications and therefore its running time does not reveal any information about the processed data.

\subsubsection{Root-Finding Step}

Root-finding can be performed by a modification of the Roth--Ruckenstein algorithm \cite{RothRuckenstein2000}.
The algorithm is outlined in Algorithm~\ref{alg:RR}.

\printalgoIEEE{
\DontPrintSemicolon
\KwIn{$M(x,y)=\smallsum_{\eta=0}^{\ell} M_\eta(x) y^\eta$, $\gcandidate(x)$, $i$, global list $\List$}
\lIf{i=k}{\Return}
$M(x,y) \gets Q(x,y)/x^r$ with $r \in \NN$ maximal \;
$p(y) \gets M(0,y)$ \;
Find roots of $p(y)$ \;
Remove $\gcandidate(x)$ from the global list $\List$ \;
\For{each root $\gamma$}{
Add $\gcandidate(x)+\gamma x^i$ to the global list $\List$ \;
$\mathrm{RR}\left(M(x,x(y-\gamma)), \gcandidate(x)+\gamma x^i, i+1, \List\right)$ \;
}
\caption{$\mathrm{RR}\left( M(x,y), \gcandidate(x), i, \List \right)$}
\label{alg:RR}
}

We need to modify the algorithm slightly as follows:
\begin{itemize}
\item We compute the $i$-th recursion step of all recursive calls before starting to compute the $(i+1)$-th recursion depth.
\item After finishing all recursion steps at depth $i$, fill the list $\List$ with random univariate polynomials of degree $\leq i$ such that the list always contains $\ell(k-1)$ polynomials and mark them as \emph{random} within the global list $\List$.
\item At depth $(i+1)$, $\mathrm{RR}$ is called for all elements of $\List$ with the corresponding bivariate polynomial $M(x,x(y-\gamma))$. If the element is random, also call the algorithm with a random bivariate polynomial of $y$-degree $\leq \ell$ but do not save the results in $\List$.
\end{itemize}
The output $\List$ of the modified algorithm without the \emph{random} entries equals exactly the output of the Roth--Ruckenstein algorithm, so its correctness follows.

\begin{theorem}
Consider Algorithm~\ref{alg:RR} with above modifications. $\mathrm{RR}\left( Q(x,y), 0, 0, \{0\} \right)$ calls $\mathrm{RR}(\cdot)$ exactly $\ell^2(k-1)$ times.
\end{theorem}

\begin{proof}
The original Roth--Ruckenstein algorithm calls itself $\leq \ell(k-1)$ times \cite{RothRuckenstein2000}, so the number of \emph{non-random} entries of $\List$ will never be $\geq \ell(k-1)$.
At recursion depth $i$, for $i=1,\dots,k$, $\mathrm{RR}(\cdot)$ is called exactly $\ell(k-1)$ times since $|\List|=\ell(k-1)$.
\end{proof}

\begin{theorem}
The number of multiplications and additions needed by Algorithm~\ref{alg:RR} for fixed parameters is independent of $Q(x,y)$.
\end{theorem}

\begin{proof}
Due to lack of space, we only give the idea:
We know that $\deg p(y) \leq \ell$, so evaluation corresponds to $\ell+1$ multiplications and $\ell$ additions of field elements. Root finding in $p(y)$ can be done by evaluating it at all elements of $\Fq$.
In recursion depth $i$, $\deg M_\eta(x) \leq \max_\mu\{ \deg Q_\mu \}+\ell i$, so computing $M(x,x(y-\gamma))$ can be done in constant time since we can treat $M_\eta(x)$ as a polynomial of degree exactly $\max_\mu\{ \deg Q_\mu \}+\ell i$. Finding $r$ is a matter of data structures. Obtaining $Q(x,y)/x^r$ and $M(0,y)$ requires no computation.
\end{proof}

Thus, the modified Roth--Ruckenstein algorithm always performs the same number of field operations and can be considered to be resilient against timing attacks, cf.~Section~\ref{subsec:fieldop}.

\subsection{Inner Codes: Reed--Muller Codes}
\label{subsec:inner_codes}

For codes of small cardinality $\kB$, as often used as inner codes, maximum likelihood decoding can used, e.g., by finding the minimum of the Hamming distances $h_i = \dH(\c+\e,\c_i)$ of the received word $\c+\e$, with $\c \in \Bcode$ and error $\e$, to all codewords $\c_i$ for $i=1,\dots,2^\kB$.
In order to not reveal information about $\c$, the $h_i$ must be carefully computed.

Let $\pi$ be a random permutation of the indices $\{1,\dots,2^\kB\}$ and $(h_{\pi(1)},\dots,h_{\pi(2^\kB)})$ be the ordered list of Hamming distances of the received word to the permuted list of codewords. We can prove the following theorem that states that even if the ordered list of Hamming distances can be extracted by an attacker, the uncertainty of the codeword does not decrease.

\begin{theorem}
$H(\c \, | \, (h_{\pi(1)},\dots,h_{\pi(2^\kB)})) = H(\c)$.
\end{theorem}

\begin{proof}
Since $h_{\pi(i)} = \dH(\c+\e,c_{\pi(i)}) = \dH(\c'+\c+\e,\c'+\c_{\pi(i)})$ for any codeword $\c' \in \Bcode$ and we can define another permutation $\pi'$ such that $\c_{\pi'(i)} = \c'+\c_{\pi(i)}$ (adding a codeword is a bijection on the code), $h_{\pi(i)} = \dH(\c'+\c+\e,\c_{\pi'(i)})$,
so the uncertainty of choosing a codeword $\c'$ remains.
\end{proof}

\section{Conclusion}
\label{sec:conclusion}

In this paper, we have presented decoding algorithms for key reproduction using PUFs that both achieve larger decoding performance than existing ones and are resistant against side-channel attacks on the runtime.
Both, list recovery \cite{guruswami_limits_2006} and the K\"otter--Vardy algorithm \cite{koetter2003algebraic}, a soft-decision variant of the Guruswami--Sudan algorithm, promise a further large gain in decoding performance.
Investigating the capability to use them for PUFs is work in progress.
Moreover, it is necessary to prevent differential power analysis (DPA) attacks on the decoding step, e.g., by combining our methods with DPA-resistant logic styles, see~\cite{wild2015glifred} and references therein.

\section*{acknowledgement}

\noindent
We would like to thank Matthias Hiller and Vincent Immler for the valuable discussions and helpful comments on an earlier version of the paper.

\bibliographystyle{IEEEtran}
\bibliography{main}

\end{document}